\documentclass[conference]{IEEEtran}
\IEEEoverridecommandlockouts

\usepackage{amsmath,amsfonts}
\usepackage{algorithmic}
\usepackage{algorithm}
\usepackage{array}
\usepackage{textcomp}
\usepackage{stfloats}
\usepackage{url}
\usepackage{verbatim}
\usepackage{graphicx}
\usepackage{cite}
\usepackage{amssymb}

\usepackage{caption}
\usepackage{subcaption}

\usepackage{graphicx}

\DeclareMathOperator*{\esssup}{ess\,sup}

\usepackage{bm}

\renewcommand{\P}{\mathsf{P}} 			

\usepackage{amsthm}
\usepackage{amsmath}

\newtheorem{proposition}{Proposition}

\newtheorem{Cor}{Corollary}

\newtheorem{definition}{Definition}
\newtheorem{remark}{Remark}
\newtheorem{fact}{Fact}

\newcommand{\signal}[1]{{\boldsymbol{#1}}}

\newcommand{\real}{{\mathbb R}}

\newcommand{\Natural}{{\mathbb N}}
\newcommand{\refeq}[1]{(\ref{#1})}

\def\BibTeX{{\rm B\kern-.05em{\sc i\kern-.025em b}\kern-.08em
		T\kern-.1667em\lower.7ex\hbox{E}\kern-.125emX}}

\begin{document}
	
		\title{Beyond the Use-and-then-Forget (UatF) Bound: Fixed Point Algorithms for Statistical Max-Min Power Control }
	
	
	


\author{%
	\IEEEauthorblockN{%
		Renato~L.~G.~Cavalcante\IEEEauthorrefmark{1},
		Noor~Ul~Ain\IEEEauthorrefmark{2},
		Lorenzo Miretti\IEEEauthorrefmark{3},
		Slawomir Sta\'nczak\IEEEauthorrefmark{1}\IEEEauthorrefmark{2}
	}\\
	\IEEEauthorblockA{\IEEEauthorrefmark{1}\textit{Wireless Communications and Networks}, Fraunhofer Heinrich Hertz Institute, Berlin, Germany}%
	\IEEEauthorblockA{\IEEEauthorrefmark{2}\textit{Network Information Theory}, Technical University of Berlin, Berlin, Germany}%
	\IEEEauthorblockA{\IEEEauthorrefmark{3}\textit{Ericsson Research}, Germany }
			\thanks{The authors acknowledge the financial support from the Federal Ministry of Research, Technology and Space (BMFTR) in Germany, project xG-RIC (grants 16KIS2429K and 16KIS243), and from the 6G-MIRAI project, which has received funding from the Smart Networks and Services Joint Undertaking (SNS JU) under the European Union’s Horizon Europe research and innovation program under Grant Agreement No 10119236. Views and opinions expressed are, however, those of the author(s) only, and they do not necessarily reflect those of the European Union or the SNS JU (granting authority). Neither the European Union nor the granting authority can be held responsible for them.}

}
	
	\maketitle

	\begin{abstract}
		We introduce mathematical tools and fixed point algorithms for optimal statistical max-min power control in cellular and cell-less massive MIMO systems. Unlike previous studies that rely on the use-and-then-forget (UatF) lower bound on Shannon achievable (ergodic) rates, our proposed framework can deal with alternative bounds that explicitly consider perfect or imperfect channel state information (CSI) at the decoder. In doing so, we address limitations of UatF-based power control algorithms, which inherit the shortcomings of the UatF bound. For example, the UatF bound can be overly conservative: in extreme cases, under fully statistical (nonadaptive) beamforming in zero-mean channels, the UatF bound produces trivial (zero) rate bounds. It also lacks scale invariance: a realization-dependent random scaling  of the beamformers can change the bound drastically. In contrast, our framework is compatible with  information-theoretic bounds that do not suffer from the above drawbacks. We illustrate the framework by solving a max-min power control problem considering a standard bound that exploits instantaneous CSI at the decoder.
		
	\end{abstract}
	 

	\section{Introduction}
	\label{sect.intro}
	The use-and-then-forget (UatF) bound has been fundamental in the development of beamforming and power control algorithms in massive MIMO and cell-less networks because it often leads to tractable and scalable optimization problems \cite{miretti2025two,demir2021,massivemimobook,miretti2024ul,cavalcante2023,marzetta16,ain2025optimal}. Furthermore, in some scenarios of practical relevance -- e.g., when optimal beamformers are employed -- the UatF bound provides spectral efficiency predictions that closely match those obtained from tighter bounds of the Shannon ergodic rates, including the coherent decoding bound \cite{miretti2025two,ain2025optimal}.
	
Despite its broad success, the UatF bound has important limitations that restrict its general applicability. For example, it replaces the instantaneous effective channel by its mean and treats deviations as uncorrelated noise \cite[Ch.~4.2]{massivemimobook}, so it leads to trivial bounds in scenarios where the beamformer remains fixed across multiple realizations of a zero-mean channel. More specifically, under such condition, the numerator of the effective signal-to-interference-plus-noise ratio (SINR) expression in \cite[Theorem~4.4]{massivemimobook} vanishes, which in turn yields a zero lower bound on the achievable rate. Furthermore, even if the UatF bound is nontrivial, it can severely underestimate the achievable rates when simple beamformers are employed, and its predictions can change even with a simple scaling of the realization of the beamforming vectors, an operation that should not affect the spectral efficiency if the decoder is aware of the scaling factor \cite[Ch.~4.2.1]{massivemimobook}. For these reasons, there is a need to develop techniques that enable the derivation of power control algorithms based on alternative bounds. From an optimization standpoint, using known tight bounds is challenging because they often require computing the expectation of a logarithmic function involving  fractions containing random variables (beamformers, channel estimates, etc.) in both the numerator and denominator \cite[Theorem~5.1]{massivemimobook}.

	Against this background, we demonstrate that, in networks where the beamforming strategy is independent of the transmit power vector (the general case will be considered elsewhere), alternatives to the UatF bound -- e.g., the coherent decoding bound and the upper bound in \cite[Proposition~1]{miretti2025joint} -- can be employed in tractable weighted max-min optimization problems, provided we accept that expectations are approximated via Monte Carlo sampling, as also done in power control problems based on the UatF bound \cite{demir2021,massivemimobook,miretti2025joint,ain2025optimal,miretti2022closed,miretti2024ul}. Nevertheless, we avoid additional approximations, such as replacing the expectation of a ratio with the ratio of expectations.
	
	 To further contextualize our contribution, we note that many existing power control algorithms largely fall into two classes: (i) \emph{per-sample schemes}, which solve an optimization problem for each channel sample \cite{yates95,nuzman07}, but suffer from poor scalability in distributed MIMO systems; and (ii) \emph{statistical-level schemes}, which solve a single optimization problem per channel distribution \cite{miretti2025two,demir2021}, but rely on the UatF bound and thus inherit its limitations. Our framework builds upon the foundations of both classes to address these drawbacks.  Specifically, we derive a fixed-point algorithm -- based on iterations introduced to the wireless literature in \cite{nuzman07} -- that provably converges to globally optimal solutions of weighted max-min problems. The algorithm uses an achievable-rate bound that exploits instantaneous channel state information (CSI) at the decoder, while requiring the solution of only a single optimization problem for each channel distribution.
	
	\section{Mathematical framework}
	\label{sect.preliminaries}
	 
In this section, we clarify notation, review relevant results from the literature, and establish the mathematical framework that will serve as the foundation for deriving power control schemes in the next section.

	\subsection{Notation and definitions}
	We use the convention that the set $\Natural$ of natural numbers does not include zero. The sets of nonnegative and positive reals are denoted by, respectively, $\real_{+}:=~[0,\infty[$ and $\real_{++}:=~]0,\infty[$. Unless otherwise mentioned, $\log$ stands for the natural logarithm. The elements of a vector $\signal{x} \in \real^N$ are denoted either as $(x_1, \ldots, x_N)$ or as $[x_1, \ldots, x_N]^t$, where $(\cdot)^t$ denotes the transpose operator. We adopt the latter notation whenever we need to emphasize that vectors should be interpreted as column vectors, ensuring that subsequent equations follow conventions commonly used in the wireless literature. Inequalities involving vectors are understood coordinate-wise. A norm $\|\cdot\|$ on $\real^N$ is \emph{monotone} if $(\forall\signal{x}\in\real_+^N)(\forall\signal{y}\in\real_+^N)~\signal{x}\le\signal{y}\Rightarrow \|\signal{x}\|\le\|\signal{y}\|$. The standard $l_1$, $l_2$, and $l_\infty$ norms (all of which are monotone on $\real_+^N$) are denoted by $\|\cdot\|_1$, $\|\cdot\|_2$, and $\|\cdot\|_\infty$, respectively. A sequence $(\signal{x}_n)_{n\in\Natural}$ of vectors in $\real^N$ is said to converge to $\signal{x}^\star$ if $\lim_{n\to\infty}\|\signal{x}_n-\signal{x}^\star\|=0$ for some norm $\|\cdot\|$ (and, hence, for every norm on $\real^N$ because of the equivalence of norms in finite-dimensional spaces). To avoid potential misunderstandings, we say that $f:\mathcal{X}\to\mathcal{Y}$, where $\mathcal{X}\subset \real^N$ and $\mathcal{Y}\subset \real^M$, is a \emph{function} only if $M=1$ or a \emph{mapping} if $M\ge 1$ (NOTE: for $M=1$, then $f$ is both a function and a mapping). The set of fixed points of a mapping $f:\mathcal{X}\to\mathcal{Y}$, where $\mathcal{Y}\subset\mathcal{X}$, is denoted by $\mathrm{Fix}(f):=\{\signal{x}\in\mathcal{X}\mid f(\signal{x})=\signal{x}\}$. A mapping $f:\real^N\to\real^M$ is concave (respectively, convex) if each coordinate function is concave (respectively, convex).   
			
	The following technical result can be established with arguments analogous to those in \cite[Lemmas~2 and 3]{renato2016power}\cite{cai2012optimal}, in combination with the standard result shown in  \cite[Theorem 10.3]{rock70}\cite[Proposition~1]{cavalcante2016}.  
	
	\begin{fact}
		\label{fact.concavity}
	    Fix $c >0$, $\sigma>0$, and $\signal{b}\in\real_+^N$. Then the function 
	    \begin{align*} 
	    	f:\real^N_{++}\to\real_{++}:\signal{x}=(x_1,\ldots,x_N)\mapsto \dfrac{x_1}{\log(1+s(\signal{x}))},
	    \end{align*} 
	    where $s:\real_+^N\to\real_+:\signal{x}\mapsto{c x_1}/{(\signal{b}^t\signal{x}+ \sigma)}$, is concave. Furthermore, $f$ can be continuously extended to the domain $\real_+^N$, and this continuous extension, denoted by $\bar{f}:\real_+^N\to\real_+$, is a concave function satisfying $(\forall\signal{x}\in\real_+^N)~\bar{f}(\signal{x})\ge \bar{f}(\signal{0})=\sigma/c>0$.  
	\end{fact}
	
	A probability space is denoted by $(\Omega, \mathcal{F}, \P)$, where $\Omega$ is the sample space, $\mathcal{F}$ is the event space, and $\P$ is the probability measure. The next result is a direct implication of \cite[Proposition~2.16]{folland1999real}.
	
	\begin{fact}
		\label{fact.positivity}
		Given a probability space $(\Omega, \mathcal{F}, \P)$ and a scalar $\delta\ge 0$, assume that a random variable $X:\Omega\to [0,\infty[$ satisfies $X>\delta$ ($\P$-)almost surely (a.s.) (or, equivalently, $X(\omega)>\delta$ for almost every $\omega\in\Omega$) and $E[X(\omega)]:=\int X(\omega) \P(\mathrm{d}\omega)$ is well-defined (i.e., $E[X(\omega)]\in\real_+$).\footnote{We adopt the stochastic-programming notation $E[X(\omega)]$ rather than $E[X]$ to avoid ambiguity when deterministic decision variables are present. This distinction is especially important later in \refeq{eq.utility} and \refeq{eq.su}, where the notation makes explicit that the power vector is not a decision variable for the beamformers, which are random vectors. } Then $E[X(\omega)]>\delta$.
	\end{fact}
	We can now proceed with the contributions of this section.

	\subsection{Framework for Power Control with Tight Ergodic Rate Bounds}
The remainder of this section has two main objectives. First, we demonstrate that the results in \cite{nuzman07} continue to hold for a slight variation of the concept of standard interference mappings originally introduced in the wireless literature in \cite{yates95} (see Proposition~\ref{prop.nuzman}). Second, we establish a connection between these results and a general class of functions that includes, as special cases, those commonly arising in weighted max-min optimization problems involving tight bounds on Shannon (ergodic) rates (Proposition~\ref{proposition.MSP}). This latter result lays the groundwork for the subsequent sections, which explore further connections to concrete problems in the wireless domain. We begin by introducing the following variant of standard interference mappings:
	\begin{definition}
		\label{def.sif}
		We say that $f:\real^N_{++}\to\real_{++}$ is a monotonic, scalable, and positive (MSP) function if the following properties hold:
		\begin{itemize}
			\item[(i)] [monotonicity] $(\forall \signal{x}\in\real_{++}^N)(\forall \signal{y}\in\real_{++}^N)~ \signal{x}\le\signal{y}\Rightarrow f(\signal{x})\le f(\signal{y})$; 
			\item[(ii)] [scalability] $(\forall \alpha>1)(\forall \signal{x}\in\real_{++}^N)~ f(\alpha\signal{x})<\alpha f(\signal{x})$; and
			\item[(iii)] [positivity (bounded away from $0$)] $\inf_{\signal{x}\in\real_{++}^N}f(\signal{x})>0$.
			
		\end{itemize}
		Likewise, we say that a mapping $T:\real^N_{++}\to\real_{++}^N:\signal{x}\mapsto(f_1(\signal{x}),\ldots,f_N(\signal{x}))$ is an MSP mapping if the coordinate functions $(f_n:\real_{++}^N\to\real_{++})_{n\in\{1,\ldots,N\}}$ are MSP functions.
	\end{definition}

	The key distinction between the functions in Definition~\ref{def.sif} and the standard interference functions in  \cite{yates95} lies in their respective domains. Standard interference functions restricted to the positive cone $\real_{++}^N$ are MSP functions by definition, and both classes of functions are continuous on $\real_{++}^N$ \cite{burbanks2003extension}. However, the continuous extensions of MSP functions to the nonnegative cone $\real_{+}^N$, which are guaranteed to exist \cite{burbanks2003extension}\cite[Theorem~5.1.5]{lem13}, are not standard interference functions in general. Without additional assumptions, these existing results and the positivity property in Definition~\ref{def.sif} guarantee only the preservation of positivity, monotonicity, and weak scalability on $\real_{+}^N\setminus \real_{++}^N$-- i.e., the strict inequality in the scalability condition has to be relaxed to a weak inequality in general. For example, consider the continuous extension $\bar{f}:\real^2_{+}\to\real_{++}$ of the function $f:\real^2_{++}\to\real_{++}:(x_1,x_2)\mapsto 1+x_1+\sqrt{x_2}+\max\{x_1-1,0\}$ to the domain $\real_{+}^2$, which is just the function $f$ itself with its domain extended to $\real_{+}^2$. It satisfies all properties in Definition \ref{def.sif} on $\real_{++}^2$, but, for $\signal{x}=(1,0)\notin\real_{++}^2$, we do not have scalability because $(\forall\alpha>1)~\bar{f}(\alpha\signal{x})=\alpha \bar{f}(\signal{x})$.

In Definition~\ref{def.sif} we exclude the boundary of the nonnegative cone to avoid unnecessary notational clutter in the applications we consider. To illustrate this point, consider the function $f$ in Fact~\ref{fact.concavity} for $N=1$ and $\signal{b}=0$, which is an MSP function, so it can be continuously extended to the boundary $\real_+\setminus \real_{++}$ of its original domain. From a notational standpoint, mathematical rigor requires the extension $\bar{f}:\real_+\to\real_{+}$ of $f$ to be written as $(\forall x\ge0)~\bar{f}(x):=\begin{cases} f(x),\text{if } x_1 > 0 \\ \sigma/c_1\text{ otherwise} \end{cases}$ to avoid divisions by zero. Another subtlety is that the definition in \cite{yates95} does not require continuity on the boundary of the domain \cite{cavalcante2016}. For instance, the function on $\real_+$ defined by $f(x)=2$ if $x>0$ and $f(0)=1$ is a standard interference function in the sense defined in \cite{yates95}. With Definition~\ref{def.sif}, we relegate boundary issues to the proofs. We also note that the assumption of positivity in the original definition of standard interference functions in \cite{yates95} is known to be redundant -- it follows from the other properties of the definition \cite{leung2004}. In contrast, the assumption of positivity in Definition~\ref{def.sif}, which is essential for our proofs based on the arguments in \cite{nuzman07,cavalcante2023}, does not follow from properties (i) and (ii). For example, the function $f:\real_{++}\to\real_{++}:x\mapsto\log(1+x)$ satisfies properties (i) and (ii), but it fails to satisfy property (iii) because $\inf_{x>0}f(x)=\lim_{x\to 0^+} f(x) = 0$.

Although the above aspects might prevent the direct application of some results for standard interference functions to MSP functions -- because, as discussed above, the set of standard interference functions restricted to the positive cone is a proper subset of MSP functions\footnote{In turn, MSP functions are a proper subclass of order-preserving subhomogeneous functions \cite{lem13,krause2015positive}.} --, many key properties remain analogous. In particular, the class of MSP functions is closed under addition, positive scaling, and finite maxima and minima, just as in the case of standard interference functions. Furthermore, any MSP mapping $T:\real_{++}^N\to\real_{++}^N$ is contractive in the complete metric space $(\real_{++}^N, d)$ \cite[Lemma~2.1.7]{lem13} (but not necessarily a Lipschitz contraction), where $d:\real_{++}^N\times\real_{++}^N\to\real_{+}$ is Thompson's metric defined in \cite[p.~30]{lem13}, so $T$ has at most one fixed point. If $\mathrm{Fix}(T)\neq\emptyset$,  arguments similar to those described in \cite{yates95} to prove convergence of the sequence $(\signal{x}_n)_{n\in\Natural}$ generated via the fixed point iterations $(\forall n\in\Natural)~\signal{x}_{n+1}=T(\signal{x}_n)$, with $\signal{x}_1\in\real_{++}^N$, can be used for MSP mappings. In addition, the following property, which is crucial for the power control algorithms derived in the next section, can be established using similar arguments to those in \cite{nuzman07}, and they are therefore omitted because of the space limitation.

\begin{proposition}
	\label{prop.nuzman} 
	Let $T:\real_{++}^N\to\real_{++}^N$ be an MSP mapping. Given a monotone norm $\|\cdot\|$, there exists exactly one tuple $(\gamma^\star,\signal{x}^\star)\in\real_{++}\times\real_{++}^N$ solving the so-called conditional (nonlinear) eigenvalue problem:
	\begin{align}
		\label{eq.condeig}
		\text{ Find } (\gamma,\signal{x})\in\real_{++}\times\real_{++}^N \text{ s.t. }T(\signal{x})=\gamma\signal{x},~\|\signal{x}\|=1. 
	\end{align}
	Furthermore, the solution $(\gamma^\star,\signal{x}^\star)$ to \refeq{eq.condeig} (called the conditional eigenpair) is such that the conditional eigenvalue $\gamma^\star$ satisfies $\gamma^\star=\|T(\signal{x}^\star)\|$, and the conditional eigenvector $\signal{x}^\star$ is the unique fixed point $\signal{x}^\star\in\mathrm{Fix}(G)$ of the mapping $G:\real_{++}^N\to\real_{++}^N:\signal{x}\mapsto (1/\|T(\signal{x})\|) T(\signal{x})$. In addition, given $\signal{x}_1\in\real_{++}^N$, the sequence $(\signal{x}_n)_{n\in\Natural}$ generated via $(\forall n\in\Natural)~\signal{x}_{n+1} = G(\signal{x}_n)$ converges to $\signal{x}^\star$.
\end{proposition}

The above result is especially relevant to the following class of functions:

\begin{proposition}
	\label{proposition.MSP}
	Let $(\Omega, \mathcal{F}, \P)$ be a probability space, and let $g:\real_{++}^N\times\Omega\to\real_{+}$ and $h:\real_{++}^N\to\real_{++}$ be functions such that, for every $\signal{x}\in\real_{++}^N$, the function $\omega\mapsto g(\signal{x},\omega)$ is measurable. Assume that there exists a set $\Omega_0\in\mathcal{F}$, with $\P(\Omega_0)=1$, and a random variable $\delta:\Omega\to\real_{++}$ such that $E[1/\delta(\omega)]<\infty$. Further assume that, for every $\omega\in\Omega_0$, $g(\signal{x},\omega)>0$ for every $\signal{x}\in\real_{++}^N$, the function 
	 $q_\omega:\real_{++}^N\to\real_{++}:\signal{x}\mapsto h(\signal{x})/g(\signal{x},\omega)$ is an MSP function, and $q_\omega(\signal{x})\ge \delta(\omega)$ for every $\signal{x}\in\real_{++}^N$. Then $(\forall\signal{x}\in\real_{++}^N)~ 0<E[g(\signal{x},\omega)]<\infty$ and the function
	\begin{align}
		f:\real^N_{++}\to\real_{++}:\signal{x}\mapsto\dfrac{h(\signal{x})}{E[g(\signal{x},\omega)]}
	\end{align}
	is an MSP function.
\end{proposition}
\begin{proof}
Fix $\signal{x}\in\real_{++}^N$. From the assumptions of the proposition, $0<g(\signal{x},\omega)/h(\signal{x})\le 1/\delta(\omega)$ (a.s.), and thus $0<E[g(\signal{x},\omega)]\le h(\signal{x})E[(\delta(\omega))^{-1}]<\infty$ as a direct application of Fact~\ref{fact.positivity}. Therefore,  $(\forall\signal{x}\in\real_{++}^N)~ 0<E[g(\signal{x},\omega)]<\infty$. In particular, this result shows that $f$ is well defined with codomain in $\real_{++}$. We now proceed to prove monotonicity, scalability, and positivity of $f$.
		
	(Monotonicity) Fix two vectors $\signal{x}\in\real_{++}^N$ and $\signal{y}\in\real_{++}^N$ such that $\signal{x}\le\signal{y}$. For every $\omega\in\Omega_0$, $q_\omega$ is an MSP function by assumption, and thus 
	\begin{align*}
		\dfrac{g(\signal{x},\omega)}{h(\signal{x})} \ge  \dfrac{g(\signal{y},\omega)}{h(\signal{y})}>0.
	\end{align*}
	Monotonicity of expectations and Fact~\ref{fact.positivity} yield $1/f(\signal{x})=E[g(\signal{x},\omega)]/h(\signal{x})\ge 1/f(\signal{y})=E[g(\signal{y},\omega)]/h(\signal{y})>0$, which implies monotonicity of $f$. \\

	(Scalability) Fix $\alpha>1$ and $\signal{x}\in\real_{++}^N$ arbitrarily. Using scalability and positivity of $q_\omega$ for every $\omega\in\Omega_0$, we deduce
	\begin{align*}
	\dfrac{g(\alpha \signal{x},\omega)}{h(\alpha \signal{x})} > \dfrac{g(\signal{x},\omega)}{\alpha h(\signal{x})}>0.
	\end{align*}
		Taking expectations on both sides of the inequality, and using Fact~\ref{fact.positivity} by setting $X(\omega)= g(\alpha\signal{x},\omega)/h(\alpha\signal{x})-g(\signal{x},\omega)/(\alpha~h(\signal{x}))>0$ if $\omega\in\Omega_0$ and $X(\omega)=0$ otherwise (NOTE: $X:\Omega\to\real_+$ is measurable with finite expectation because $0\le X(\omega)\le 1/\delta(\omega)$ a.s.), we obtain $1/f(\alpha\signal{x})=E[g(\alpha\signal{x},\omega)]/h(\alpha\signal{x})>E[g(\signal{x},\omega)]/(\alpha h(\signal{x})) =1/(\alpha~f(\signal{x}))>0$, which implies scalability of $f$.
	
	(Positivity) Fix $\signal{x}\in\real_{++}^N$. For every $\omega\in\Omega_0$, we have $0< g(\signal{x},\omega)/h(\signal{x})\le 1/\delta(\omega)$ from the assumptions of the proposition. Fact~\ref{fact.positivity}, monotonicity of expectations, and the assumption $E[1/\delta(\omega)]<\infty$ yield	
		\begin{align*}
				0<E[g(\signal{x},\omega)/h(\signal{x})]\le E[1/\delta(\omega)]<\infty.
        \end{align*}
        Consequently, $f(\signal{x})=(E[g(\signal{x},\omega)/h(\signal{x})])^{-1}\ge (E[1/\delta(\omega)])^{-1}=:c >0$. Since $\signal{x}\in\real_{++}^N$ is arbitrary and $c$ is independent of $\signal{x}$, we deduce $\inf_{\signal{x}\in\real_{++}^N}f(\signal{x})\ge c >0,$ which is the desired result.
	\end{proof}

\section{Power control in multi-user MIMO}

We consider the uplink of a general multi-user MIMO network with $L\in\Natural$ access points, each equipped with $M\in\Natural$ antennas. Thus, the total number of antennas on the access point side is $K:=LM$. The set of access points is denoted by $\mathcal{A}:=\{1,\ldots,L\}$. The system also includes $N\in\Natural$ single-antenna users, which are indexed by the set $\mathcal{U}:=\{1,\ldots,N\}$. The channel and beamforming vectors for user $u\in\mathcal{U}$ and access point $a\in\mathcal{A}$ are random (column) vectors denoted by, respectively, $\signal{h}_{u,a}:\Omega\to\mathbb{C}^M$ and $\signal{v}_{u,a}:\Omega\to\mathbb{C}^M$, where $(\Omega,\mathcal{F},\P)$ is the probability space as defined below Fact~\ref{fact.concavity}. The transmit power of user $u\in\mathcal{U}$ is denoted by $p_u\in\real_{++}$. Without any loss of generality, we assume that $(\forall\omega\in\Omega)~ \signal{v}_{u,a}(\omega)=\signal{0}$ if access point $a\in\mathcal{A}$ does not participate in the processing of the signal transmitted by user $u\in\mathcal{U}$, and that $\|\signal{v}_u(\omega)\|_2=1$ for almost every $\omega\in\Omega$ and every $u\in\mathcal{U}$. The aggregated channel vector and aggregated beamforming vector for user $u\in\mathcal{U}$ are denoted by, respectively, $\signal{h}_u:\Omega\to\mathbb{C}^{K}:\omega\mapsto[\signal{h}_{u,1}^t(\omega),\ldots,\signal{h}^t_{u,L}(\omega)]^t$ and $\signal{v}_u:\Omega\to\mathbb{C}^{K}:\omega\mapsto[\signal{v}^t_{u,1}(\omega),\ldots,\signal{v}^t_{u,L}(\omega)]^t$.

The above model covers many common technologies in the literature, including multi-cell massive MIMO networks as defined in \cite{massivemimobook} and  user-centric cell-less networks with different levels of cooperation as defined in \cite{demir2021}. For brevity and notational simplicity (see Remark~\ref{remark.bound} below), we evaluate the utility (performance) of user $u\in\mathcal{U}$ as a function of the power allocation $\signal{p}=[p_1,\ldots,p_N]^t\in\real_{++}^N$ via the following well-known achievable rate [in nats/s/Hz] under perfect CSI \cite[Proposition~1]{miretti2025joint}\cite[Lemma~1]{caire2018ergodic}:
\begin{align}
	\label{eq.utility}
	r_u:\real^N_{++}\to\real_+:\signal{p}\mapsto E\left[\log(1+s_u(\signal{p},\omega))\right],
\end{align}
where 
\begin{align}
	\label{eq.su}
	\begin{array}{rcl}
		s_u:\real_{++}^N\times\Omega&\to&\real_+\\ (\signal{p},\omega)&\mapsto&\dfrac{p_u |\signal{h}_u(\omega)^H\signal{v}_u(\omega)|^2}{\sum_{k\in\mathcal{U}\setminus\{u\}} p_k|\signal{h}_k(\omega)^H\signal{v}_u(\omega)|^2 + \sigma^2},
	\end{array}
\end{align}
can be interpreted as the instantaneous SINR function and $\sigma^2\in\real_{++}$ is the noise power. Important aspects of the model we consider are discussed below.

\begin{remark}\label{remark:beamformers}
	The beamformers $(\signal{v}_u:\Omega\to\mathbb{C}^K)_{u\in\mathcal{U}}$ in \refeq{eq.su} are random vectors, so they are allowed to change for each channel realization. However, as clear from the notation, they should not depend on the power vector $\signal{p}\in\real_+^N$, the argument of the functions $(r_u)_{u\in\mathcal{U}}$. This assumption is able to cover many prominent beamforming strategies such as conventional conjugate beamforming, (regularized) zero-forcing beamforming, and MMSE-like beamformers such as those used in \cite{miretti2022closed,demir2021,massivemimobook}. We can lift this assumption within our framework by replacing $r_u$ in \refeq{eq.utility} with the following rate function: 
	$$\tilde{r}_u:\real^N_{++}\to\real_+:\signal{p}\mapsto  E\left[\esssup_{\signal{v}\in\mathcal{V}_u}\log(1+\tilde{s}_u(\signal{p},\signal{v}(\omega),\omega))\right],$$ 
	where $\mathcal{V}_u$ denotes the set of admissible (possibly local) beamformers for user $u\in\mathcal{U}$, and $\tilde{s}_u(\signal{p},\signal{v}(\omega),\omega)$ is the instantaneous SINR for a given power vector $\signal{p}\in\real_+^N$, beamformer $\signal{v}:\Omega\to\mathbb{C}^K$, event $\omega\in\Omega$, and user $u\in\mathcal{U}$. 
	For brevity, and owing to space limitations, we leave joint power control and beamforming schemes to future work. Nevertheless, the modifications outlined in this remark provide a brief sketch of how to derive provably optimal algorithms for two-timescale max-min problems based on bounds other than the UatF bound.
\end{remark}

\begin{remark}
	\label{remark.montecarlo}
The algorithms derived later require the computation of the expectation in \refeq{eq.utility}. In practice, this expectation can be approximated via Monte Carlo sampling, which requires statistical knowledge of the system. Although this assumption may appear stringent, statistical information, whether complete or partial, and Monte Carlo sampling is used by all state-of-the-art power control schemes based on the UatF bound whenever existing channel models are involved and additional approximations are avoided \cite{demir2021,massivemimobook,miretti2025joint,ain2025optimal,miretti2022closed,miretti2024ul}. Hence, our assumptions are {no stronger than those commonly adopted in the literature.} We simply state them explicitly. Furthermore, in many applications that do not require real-time operation, the relevant expectations can be estimated accurately via empirical averages over a large number of samples. Examples include network planning, the generation of datasets for training neural networks for power control \cite{chafaa2025}, and the benchmarking of heuristic methods. However, Sect.~\ref{sect.simulations} provides empirical evidence that good accuracy can be achieved with relatively few samples (on the order of a few hundred or less), mirroring widely reported results for UatF-based power allocation schemes that rely on Monte Carlo sampling.   Moreover, Monte Carlo sampling is highly amenable to parallelization, so we do not rule out real-time deployment. In particular, for online schemes that iteratively update transmit powers during data transmission, SINR samples can be obtained directly from pilot measurements -- for instance, from demodulation reference signals (DMRS) in 5G. 
\end{remark}

\begin{remark}
	\label{remark.bound}
The utility in \refeq{eq.utility} is often called the ``genie-aided'' or ``optimistic'' upper bound for the achievable ergodic rates (when interference is treated as noise) because it assumes perfect CSI at the receiver \cite[p. 328]{demir2021}. Nevertheless, the results we derive can be readily extended to other  bounds, such as the coherent decoding lower bound in \cite[Theorem~5.1]{demir2021}\cite[Theorem~4.1]{massivemimobook}. While that bound uses channel estimates -- instead of perfect channels -- and  introduces an ``error covariance'' term in the denominator of the expectation in \refeq{eq.utility}, the subsequent analysis in this study applies to it in a line-by-line manner with only straightforward changes.
\end{remark}

We now establish basic properties of the utility functions $(r_u)_{u\in\mathcal{U}}$ that will allow us to derive power control algorithms based on the fixed point iterations described in Proposition~\ref{prop.nuzman}. 

For every user $u\in\mathcal{U}$, we assume that $r_u(\signal{p})>0$ if $\signal{p}\in\real_{++}^N$  [$r_u$ is defined in \refeq{eq.utility}]. Otherwise, the user can be excluded from consideration because such a user cannot be served by the network. Therefore, denoting by $\gamma_u\in\real_+$ the utility of user $u\in\mathcal{U}$ for a given power $\signal{p}\in\real_{++}^N$, we deduce 
\begin{align*}
	0<\gamma_u = r_u(\signal{p}) \Leftrightarrow p_u={\gamma_u} f_u(\signal{p}),
\end{align*}
where 
\begin{align}
	\label{eq.msp}
	(\forall u\in\mathcal{U})~f_u:\real_{++}^N\to\real_{++}:\signal{p}\mapsto \dfrac{p_u}{r_u(\signal{p})}.
\end{align}
The next result establishes a fundamental connection among the utility in \refeq{eq.utility}, the function in \refeq{eq.msp}, and Proposition~\ref{proposition.MSP}:

\begin{proposition}
	\label{prop.utility}
	For every user $u\in\mathcal{U}$, assume that the expectation in \refeq{eq.utility} is well defined for every power vector $\signal{p}\in\real^N_{++}$, and $\omega \mapsto \signal{h}_u(\omega)^H\signal{v}_u(\omega)$ is a complex-valued random variable satisfying $E[|\signal{h}_u(\omega)^H\signal{v}_u(\omega)|^2]<\infty$ and $|\signal{h}_u(\omega)^H\signal{v}_u(\omega)|\ne 0$ for almost every $\omega\in\Omega$.  Then  $f_u:\real^N_{++}\to\real_{++}$ in \refeq{eq.msp} is an MSP function for every $u\in\mathcal{U}$.
\end{proposition}
\begin{proof}
	Fix $u\in\mathcal{U}$ arbitrarily. Fact~\ref{fact.concavity} and the assumption $|\signal{h}_u(\omega)^H\signal{v}_u(\omega)|\ne 0$ (a.s.) show that, for every $\omega\in\Omega_0:=\{\omega\in\Omega\mid |\signal{h}_u(\omega)^H\signal{v}_u(\omega)|\ne 0\}$, the function 
	\begin{align*} 
		q_\omega:\real_{++}^N\to\real_{++}:\signal{p}\mapsto p_u/\log(1+s_u(\signal{p},\omega)),
	\end{align*} 
	is concave [$s_u$ is defined in \refeq{eq.su}], and it has a continuous concave extension $\bar{q}_\omega:\real_{+}^N\to\real_{++}$ satisfying $(\forall\signal{p}\in\real_{++}^N)~q_\omega(\signal{p})\ge \bar{q}_\omega(\signal{0})=\sigma^2/|\signal{h}_u(\omega)^H\signal{v}_u(\omega)|^2>0$. Now apply \cite[Proposition~1]{cavalcante2016} to conclude that $\bar{q}_\omega:\real_+^N\to\real_{++}$ is a standard interference function as defined in \cite{yates95}, so it also satisfies all properties in Definition~\ref{def.sif}. As a result, the function ${q}_\omega:\real_{++}^N\to\real_{++}:\signal{p}\mapsto \bar{q}_\omega(\signal{p})$, which is the restriction of $\bar{q}_\omega$ to the domain $\real_{++}^N$, is an MSP function. Defining $\delta(\omega):=\sigma^2/|\signal{h}_u(\omega)^H\signal{v}_u(\omega)|^2$ if $\omega\in\Omega_0$ or $\delta(\omega)=1$ otherwise, we verify that $E[1/\delta(\omega)]=E[|\signal{h}_u(\omega)^H\signal{v}_u(\omega)|^2]/\sigma^2<\infty$. In addition, define $g(\signal{p},\omega):=\log(1+s_u(\signal{p},\omega))$ if $\omega\in\Omega_0$ or $g(\signal{p},\omega):=0$  otherwise. 
	The desired result follows from Proposition~\ref{proposition.MSP} by using $h(\signal{p}):=p_u$ for every $\signal{p}\in\real_{++}^N$ (readers can verify that all conditions of Proposition~\ref{proposition.MSP} are satisfied with the above definitions).
\end{proof}

We now have all the necessary background to pose and solve a power control problem that has the objective of maximizing the minimum (weighted) achievable rate  while satisfying the power constraint $\{\signal{p}\in\real_{++}^N\mid\|\signal{p}\|\le p_\mathrm{max}\}$, where $\|\cdot\|$ is a given monotone norm (typically the $l_\infty$-norm in the uplink) and $p_\mathrm{max}>0$ the maximum transmit power. Formally, we consider the following optimization problem:
\begin{align}
	\label{eq.original}
	\begin{array}{rl}
		\text{maximize}_{\signal{p}\in\real_{++}^N} &\min_{u\in\mathcal{U}} \alpha_u^{-1}~r_u(\signal{p}) \\
		\text{s.t.} & \|\signal{p}\|\le p_\textrm{max},
	\end{array}
\end{align}
where $(\alpha_1,\dots,\alpha_N)\in\real_{++}^N$ corresponds to the  weights (priorities) assigned to users, and $r_u$ is the achievable rate of user $u\in \mathcal{U}$ in \refeq{eq.utility}. In light of \refeq{eq.msp}, the epigraph form of the above problem is given by:
\begin{align}
	\label{eq.epigraph}
	\begin{array}{rl}
\text{maximize}_{(\gamma,\signal{p})\in\real_{++}\times\real_{++}^N} & \gamma \\
\text{s.t.} & \signal{p} \ge \gamma T(\signal{p}) \\
& \|\signal{p}\|\le p_\textrm{max},
\end{array}
\end{align}
where $T:\real_{++}^N\mapsto \real_{++}^N:\signal{p}\mapsto [\alpha_1 f_1(\signal{p}),\ldots,\alpha_N f_N(\signal{p})]^t$, with $(f_u)_{u\in\mathcal{U}}$ as defined in \refeq{eq.msp}, is an MSP mapping as an immediate consequence of Proposition~\ref{prop.utility}.

We can use the same arguments in \cite[Proposition~2]{cavalcante2023} with standard interference functions replaced with MSP functions to show that, although the solution to \refeq{eq.epigraph} (and, hence, \refeq{eq.original}) may not be unique, there always exists a solution achieving equality in all constraints; i.e., there exists a solution $(\gamma^\star,\signal{p}^\star)$ satisfying $(\gamma^\star)^{-1} \signal{p}^\star= T(\signal{p}^\star)$ and $(1/p_\text{max})\|\signal{p}^\star\|=1$. These equalities establish that $(\gamma^\star)^{-1}$ and $\signal{p}^\star$ are, respectively, the conditional eigenvalue and conditional eigenvector of the mapping $T$ and the monotone norm $(\forall\signal{p}\in\real^N)~ \|\signal{p}\|_\star :=  (1/p_\text{max})\|\signal{p}\|$. Therefore, as a corollary of Proposition~\ref{prop.nuzman}, Proposition~\ref{prop.utility}, and \cite[Proposition~2]{cavalcante2023} (after adapting the proof of this last proposition by substituting standard interference functions with MSP functions), we obtain the following result, which shows a simple fixed point algorithm that not only solves Problem~\refeq{eq.original}, but it also converges to a particularly desirable solution when multiple solutions exist: the iterates converge to the solution with minimum total transmit power.
\begin{Cor}
	\label{cor.fpi}
	Let the assumptions in Proposition~\ref{prop.utility} be valid. Denote by $\mathcal{S}\neq\emptyset$ the set of solutions to \refeq{eq.original} and let $\|\cdot\|$ be the monotone norm in \refeq{eq.original}. Given $\signal{p}_1\in\real_{++}^N$, the sequence  $(\signal{p}_n)_{n\in\Natural}$ generated via $(\forall n\in\Natural)~ \signal{p}_{n+1}=(p_\mathrm{max}/\|T(\signal{p}_n)\|) T(\signal{p}_n)$ converges to the unique power vector $\signal{p}^\star\in\real_{++}^N$ that solves the following optimization problem:  $$\mathrm{minimize}_{\signal{p}=(p_1,\ldots,p_N)\in\mathcal{S}} \sum_{n=1}^N p_n.$$
\end{Cor}

\section{Simulations}
\label{sect.simulations}

We consider a cell-less network with $L=16$ access points, each equipped with a four-element uniform linear array with half-wavelength spacing ($\lambda/2$), where $\lambda=c/f_{\mathrm{c}}$ is the wavelength, $f_{\mathrm{c}}$ is the carrier frequency, and $c$ is the propagation speed. The access points are placed uniformly over a $1000\times 1000~\mathrm{m}^2$ area. We distribute $N=25$ single-antenna users uniformly at random and apply a conventional wrap-around technique to emulate an unbounded service area. The user-access point antenna height difference is $\Delta h=11$ m. The large-scale fading (path loss) $\beta_{u,a}$ between user $u\in\mathcal{U}$ and access point $a\in\mathcal{A}$ follows the COST-231 Walfish–Ikegami UMi model \cite[Sect.~5.2]{3gpp}:
\begin{align*} \beta_{u,a} = -35.4 + 20\log_{10}(f_\mathrm{c}) + 26\log_{10}\left(\frac{d_{u,a}}{1~\text{m}}\right) + F_{u,a}~[\text{dB}], 
\end{align*}
where $d_{u,a}$ is the 3D distance and $F_{u,a}\sim\mathcal{N}(0,\sigma_{\mathrm{sf}}^{2})$ models shadow fading (in dB). Small-scale fading follows a spatially correlated Rician model. All beamformers operate on channel estimates -- rather than the true channels -- obtained using the estimation procedure described in \cite{ain2025optimal}. For pilot assignment and formation of user-centric cooperation clusters in the cell-less network, we use the algorithm described in \cite[Algorithm~4.1]{demir2021}. The main simulation parameters are listed in Table~\ref{tab:parameters}, and we refer the readers to \cite{ain2025optimal} for further details on the simulation.

\begin{table}[h!]
\centering
\renewcommand{\arraystretch}{1} 
\setlength{\tabcolsep}{10pt} 
\begin{tabular}{c c} 
\hline
\textbf{Parameter} & \textbf{Value} \\ [1ex]
\hline\hline
Network area & $1000$ × $1000$ $\text{m}^2$ \\
\hline
Number of APs & $L=16$\\
\hline
Number of users & $N=25$\\
\hline
Number of antennas per access point & $M=4$ \\
\hline
Bandwidth & $B=20$ MHz\\
\hline
Carrier frequency & $f_c=3.7$ GHz\\
\hline
Maximum uplink transmit power & $p_{\text{max}}= 200$ mW\\
($\|\cdot\|_\infty$ used as the  & \\ 
 monotone norm in \refeq{eq.original}) & \\
\hline
Coherence block symbols & $\tau_\mathrm{c}=200$\\
\hline
Number of pilot symbols & $\tau_\mathrm{p}=10$\\
\hline
Access point-user height difference & $\Delta h=11$ m\\
\hline
Shadow fading (Line-of-Sight) & $\sigma_{\text{sf}} = 8$ dB \\
\hline
Antenna spacing & $d=\lambda/2$\\
\hline

\end{tabular}\\
\vspace{3mm} 
\caption{Simulation parameters}
\label{tab:parameters}
\end{table}

The performance evaluation considers uniform weights in \refeq{eq.original} (i.e., $(\forall u\in\mathcal{U})~\alpha_u=1$) and centralized cell-less MMSE beamforming design with fixed (unitary) power weighting coefficients in the channel inversion stage \cite[Sect.~3-D]{miretti2022closed}. We estimate the ergodic rates by averaging across 500 Monte Carlo trials. In the proposed fixed point iterations, shown in Corollary~\ref{cor.fpi}, we use \eqref{eq.utility} as the utility function, which is called ``optimistic ergodic rate'' (OER) in the following (see Remark~\ref{remark.bound}). As the first baseline, we compute the OER bound with the solution to the standard max-min problem under the UatF bound, which admits a pseudo-closed-form expression \cite[Proposition~4]{miretti2022closed}. This baseline mirrors established practice in the literature: if simple fixed point algorithms capable of solving \refeq{eq.original} were unknown (as was the case prior to this study, to the best of our knowledge), we would consider the surrogate max-min power-control problem that uses the UatF bound instead of the target bound in \refeq{eq.utility}, with the expectation that the surrogate solution would approximate the solution to \refeq{eq.original}. As the second baseline, we report the OER bound obtained with an approach that solves the max-min problem separately for each channel realization. This deterministic problem has the same form as in \refeq{eq.original}, except that the expectation is removed from the utility in \refeq{eq.utility} and all channel-dependent quantities are treated as deterministic and fixed at their sampled values. This max-min  problem can also be solved with the result in \cite[Proposition~4]{miretti2022closed}. We estimate the corresponding OER bound by averaging the optimal instantaneous objective values over 500 Monte Carlo problem instances. 



\begin{figure}[h!]
\centering
\includegraphics[width=0.4\textwidth]{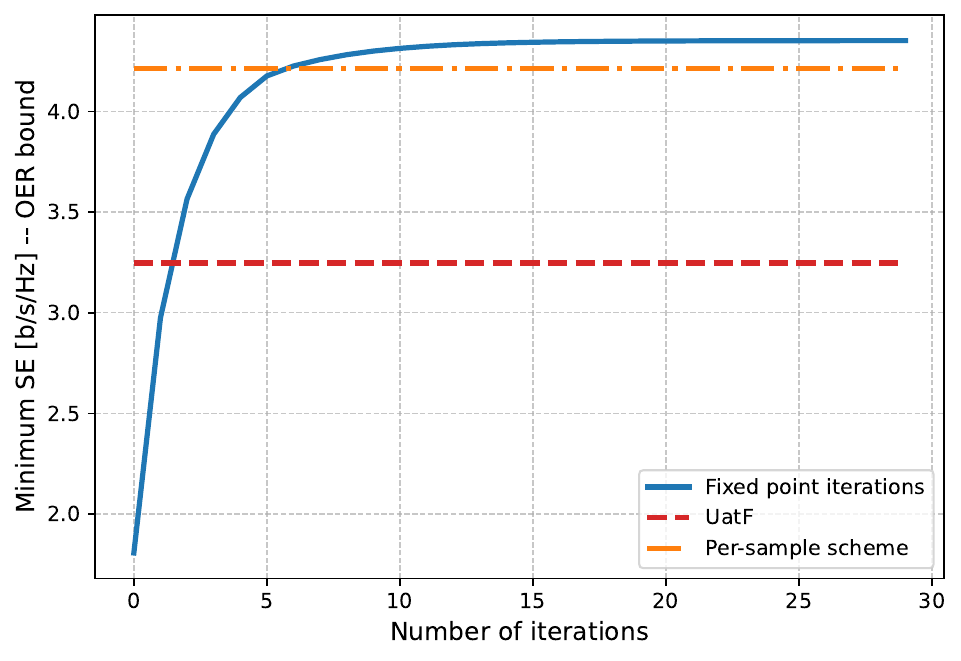}
\caption{Convergence of the worst-user rate (OER bound) achieved by the proposed iterative method under centralized MMSE beamforming. The worst-user OER values obtained from the UatF-based and per-sample max-min solutions are also shown for comparison.}
\label{fig:convergence1}
\end{figure}

Fig.~\ref{fig:convergence1} shows the convergence of the objective value for the problem in \refeq{eq.original} using the fixed point algorithm in Corollary~\ref{cor.fpi}, and, for comparison, the objective value computed with the solution to the UatF max-min problem and to the per-sample scheme. We verify that the proposed fixed point algorithm shows negligible drift from finite-sample Monte Carlo estimation of expectations, indicating a possible robustness to sampling noise, consistent with reports for state-of-the-art algorithms based on the UatF bound (see Remark~\ref{remark.montecarlo}). Furthermore, replacing the OER bound with the UatF surrogate in the max-min problem leads to a noticeable performance loss with respect to the OER bound. This result is hardly surprising: the proposed iterative algorithm provably converges to a global solution to \refeq{eq.original}, which uses the OER bound, the figure of merit considered in the simulations. Compared to the per-sample scheme, the proposed approach not only shows better performance, but it also has lower computational complexity because we solve only one optimization problem per channel distribution. A possible explanation for the underperformance of the complex per-sample scheme is that it is constrained by the instantaneous worst user, whereas the performance of the proposed algorithm is limited by the worst channel in an average sense because of the expectation of the utility in \refeq{eq.utility}. A similar phenomenon has been reported for the UatF bound in \cite{miretti2025two}. Nevertheless, to the best of our knowledge, the scenarios where this behavior occurs have yet to be formally characterized. Fig.~\ref{fig:performance2} shows the empirical cumulative distribution function of the worst-user OER bound across 200 simulations generated by dropping users uniformly at random, with all other simulation parameters set as described above. The results show that the proposed method consistently outperforms both the UatF-based approach and the per-sample scheme, consistent with the explanations given above.

\begin{figure}[h!]
	\centering
	\includegraphics[width=0.4\textwidth]{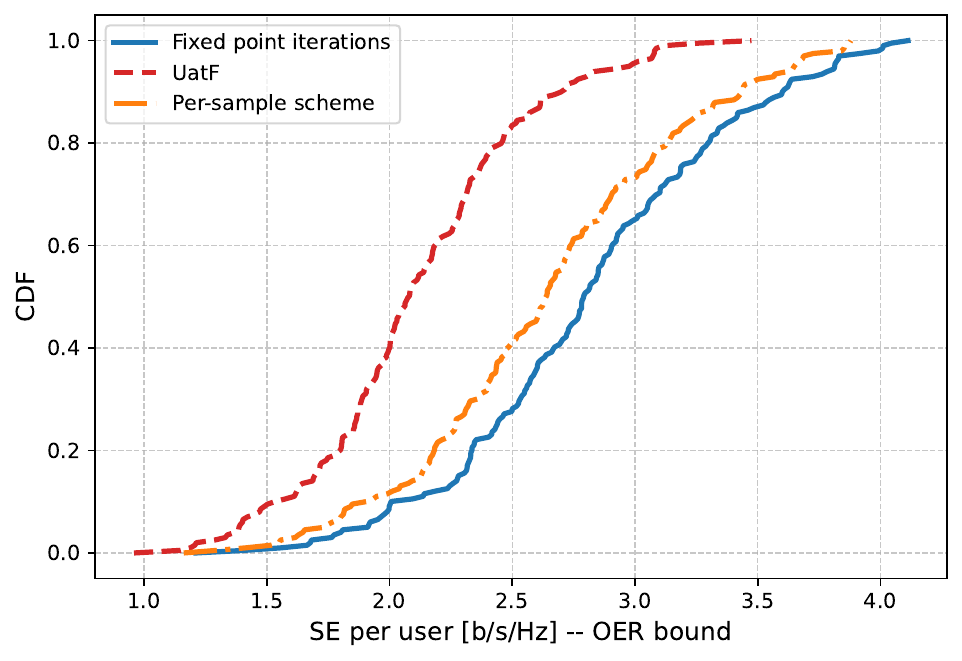}
	\caption{Empirical cumulative distribution function of the worst-user OER bound over 200 simulations with uniformly random user locations.}
	\label{fig:performance2}
\end{figure}

\section{Conclusion}
	
The UatF bound has been widely used in network optimization because dealing with alternative information-theoretic (ergodic) rate bounds were viewed as difficult. However, we have shown a framework able to target these alternative bounds directly, giving rise to simple fixed point algorithms that provably converge to global optima of max-min power control problems in cellular and cell-less massive MIMO systems, among others. Simulations show that these algorithms deliver gains in regimes where the UatF bound is overly conservative, and they can outperform complex per-sample approaches that do not exploit channel statistics and that require solving one optimization problem for each channel realization. Finally, we remark that the possible gains over algorithms based on the UatF bound can be extreme in scenarios where statistical (nonadaptive) beamforming is used under zero-mean channels, in which case the UatF bound is inapplicable because it provides trivial (zero) rate bounds.

\bibliographystyle{IEEEtran}
\bibliography{IEEEabrv,references}

\end{document}